\documentclass[11pt]{article}
\include{GrandMacros}

\usepackage{graphicx,verbatim,array,multicol, palatino}
\usepackage{amsthm, amsmath, amssymb,  setspace,natbib}
\usepackage{epsfig, url,epstopdf}
\usepackage{color}
\usepackage[super]{cite}
\usepackage{indentfirst,natbib}

\usepackage{JASA_manu}
\definecolor{red}{rgb}{1,0,0}
\definecolor{blue}{rgb}{0,0,1}
\definecolor{green}{rgb}{0,0.6,0.4}

\DeclareMathOperator{\tr}{tr}

\begin{document}

\title{Change Point Detection in Correlation Networks}

\author{Ian Barnett\thanks{Ian Barnett is a Postdoctoral Researcher and Jukka-Pekka Onnela is an Assistant Professor, Department of Biostatistics, Harvard T.H. Chan School of Public Health, 677 Huntington Avenue, Boston, MA 02115. IB is supported by PHS Grant Number 2T32ES007142-31. The authors would like to thank Dr. Stephen Maher for his review of the fMRI analysis.} \hspace{0.3pc} and \hspace{0.2pc} Jukka-Pekka Onnela\footnotemark[1]}

\date{}

\maketitle

\newpage

\begin{abstract}
Many systems of interacting elements can be conceptualized as networks, where network nodes represent the elements and network ties represent interactions between the elements. In systems where the underlying network evolves, it is useful to determine the points in time where the network structure changes significantly as these may correspond to functional change points.  We propose a method for detecting change points in correlation networks that, unlike previous change point detection methods designed for time series data, requires no distributional assumptions. We investigate the difficulty of change point detection near the boundaries of data in correlation networks and study the power of our method and a competing method through simulation. We also show the generalizable nature of the method by applying it to stock price data as well as fMRI data.

\end{abstract}

\section{Introduction}

Many systems of scientific and societal interest are composed of a large number of interacting elements, examples ranging from proteins interacting within each living cell to people interacting with one another within and across societies. These and many other systems can be conceptualized as networks, where network nodes represent the elements in a given system and network ties represent interactions between the elements. Network science and network analysis are used to analyze and model the structure of interactions in a network, an approach that is commonly motivated by the premise that network structure is associated with the dynamical behavior exhibited by the network, which in turn is expected to be associated with its function. In many cases, however, network structure is not static but instead evolves in time. This suggests that given a sequence of networks, it would be useful to determine points in time where the structure of the network changes in a non-trivial manner. Determining these points is known as the network change point detection problem. Given the connection between network structure and function, it seems reasonable to conjecture that a change in network structure may be coupled with a change in network function. Consequently, detecting structural change points for networks could be informative about functional change points as well. 

In this paper, we consider the change point detection problem for correlation networks. These networks belong to a class of networks sometimes called similarity networks and they are obtained by defining the edges based on some form of similarity or correlation measure between each pair of nodes \citep{onnela2012taxonomies}.
Examples of correlation networks appear in many financial and biological contexts, such as stock market price and gene expression data \citep{onnela2004clustering,mizuno2006correlation,bhan2002duplication,kose2001visualizing,mantegna1999hierarchical}. In general, when evaluating correlation networks, the full data is used to estimate the correlations between the nodes. When using this approach for longitudinal data, it is sometimes implied that the network structure is the same over time. This assumption may however be inaccurate in some cases. For example, in \citet{onnela2004clustering} a stock market correlation network is created from almost two decades of stock prices. In reality the relationship between the stocks, and therefore the structure of the underlying network, likely change over such a long period of time, an issue that was addressed in \citet{onnela2004clustering} by dividing the data into shorter time windows. Similarly, in functional magnetic resonance imaging (fMRI) trials it is likely that the brain interacts differently during different tasks \citep{keightley2003fmri}, or possibly even within a given task,  so it may be inaccurate to assume a constant brain activity correlation network in trials with multiple tasks.

Suppose that a network is constant or may be assumed so until a known point in time before undergoing sudden change. In this case the underlying data should be split up at the change point into two parts, and two separate correlation networks should be constructed from the two subsets of the data. In reality the location of the change point, or possibly several change points, is not known \textit{a priori} and must also be inferred from the data. This problem belongs to a wider class of so-called change point detection problems, which has been studied in the field of process control. When the observed node characteristics are independent and normally distributed, methods exist for general time series data to detect changes in the multivariate normal mean or covariance \citep{hawkins2005change,zamba2006multivariate,lowry1992multivariate}.

There have been some promising efforts at change point detection for structural networks, but in this case the actual network is observed over time rather than relying on correlations of node characteristics that are used to construct the network \citep{lindquist2007modeling,lindquist2008statistical,peel2014detecting,akoglu2010event,mcculloh2011detecting,tang2013attribute}. If a network is first inferred from correlations, then these methods could be applied. However, inferring networks from correlations is not trivial. When thresholding correlations to determine the adjacency matrix, which is a common approach, the inferred networks tend to be highly sensitive to the chosen threshold. For this reason, these methods cannot be directly applied to correlation networks without first solving the problem of inferring the correlation networks themselves. Therefore, despite this large body of methods developed for change point detection for both time series data and for networks, there is a need for a change point detection method specifically for correlation networks that is not hampered by stringent distributional assumptions.

In this paper we propose a computational framework for change point detection in correlation networks that is free from distributional assumptions. This framework offers a novel and flexible approach to change point detection. The change point detection method suggested by \citet{zamba2009multivariate,lowry1992multivariate} is adapted to our framework and its power to detect change points is compared to our method using simulation. Also, we investigate the general difficulty of change point detection near the boundaries of the data both analytically and through simulation. Finally, we apply our framework to both stock market and fMRI correlation network data and demonstrate its success and limitations for detecting functionally relevant change points.

\section{Method}

\subsection{Notation}

Assume that the system under investigation consists of a fixed number of $n$ nodes with characteristics observed at $T$ distinct time points, where the observed characteristics are $Y_{n \times T} = [Y_1,...,Y_T]$ where $Y_j = [Y_{1j},...,Y_{nj}]^T$ is the $j$th $n$-dimensional column vector of $Y$ and we assume that $Y_j \sim f(\cdot \mid \Sigma_j)$ is an unknown function with all columns of $Y$ i.i.d. (independent and identically distributed)  and where $\mbox{cov}(Y_j)=\Sigma_j$. We also assume that the rows of $Y$, corresponding to observations at individual nodes, are centered to have temporal mean $0$ and scaled to have unit variance. Note that the centering and scaling, resulting in standardized observations for each node, can always be performed.

We define a set of diagonal matrices $D(i,j)_{T \times T}$ for $1 \leq i < j \leq T$ such that
$$D(i,j)_{kk} = \left\{\begin{array}{cc}
1/(j-i+1) & \text{if } i\leq k \leq j \\
0	& \mbox{otherwise}	\\
\end{array}\right.$$

We define the covariance matrix $S(i,j)_{n \times n}$ on the subset of the data ranging from the $i$th column to the $j$th column, i.e., from time point $i$ to time point $j$ ($1 \leq i < j \leq T$), to be:

\begin{equation}\label{empcor}
S(i,j) = \sum_{k=i}^j Y_kY_k^T/(j-i+1) = YD(i,j)Y^T
\end{equation}

In order to detect a change point, we wish to find the value of $k$ in the range $(1+\Delta,T-\Delta)$  that maximizes the differences between $S(1,k)$ and $S(k+1,T)$, where $\Delta$ is picked large enough to avoid ill-conditioned covariance matrices ($\Delta > n$). The rationale for this approach is that if there were a change point in the data, the sample correlation matrices on each side of the change point ought to be different in structure. We choose the squared Frobenius norm as our metric for the distance between two matrices. Let our matrix distance metric be:
$$d(k) = ||S(1,k)-S(k+1,T)||_F^2 = \tr\{[S(1,k)-S(k+1,T)]^T[S(1,k)-S(k+1,T)]\},$$
where $\tr$ is the matrix trace operator.
We wish to test the hypotheses:
$$
\begin{array}{rl}
H_0	:& \Sigma_j = \Sigma \;\;\;\; \forall \; j	\\
H_A :& \mbox{There exists $k$ such that } \Sigma_j = \left\{\begin{array}{cc}
\Sigma_1 & \mbox{ for } j \leq k	\\
\Sigma_2 & \mbox{ for } j > k	\\
\end{array} \right.	
\end{array}
$$

\subsection{Existing methodology for change point detection}

Consider for a moment the case where the vector $Y_j$ is multivariate normal with expectation $\mu_1$ and variance-covariance matrix $\Sigma_1$ before the change point and expectation $\mu_2$ and variance-covariance $\Sigma_2$ after it. We denote this $Y_j \sim MVN(\mu_1,\Sigma_1)$ for $j\leq k$ and $Y_j \sim MVN(\mu_2,\Sigma_2)$ for $j>k$. A multivariate exponentially weighted moving average (EWMA) model has been developed for the detecting when $\mu_1$ changes to $\mu_2$ \citep{zamba2006multivariate,lowry1992multivariate}. A likelihood ratio test for detecting change points in the covariance matrix $\Sigma_1$ at a known fixed point $k$ was considered by \citet{zamba2009multivariate} and \citet{andersonw}. The likelihood ratio test statistic for detecting a change point at $k$ is

\begin{equation} \label{LRteststat}
\Lambda_k = \frac{|S(1,k)|^{\frac{k-1}{2}}\cdot |S(k+1,T)|^{\frac{T-k-1}{2}}}{|S(1,T)|^{\frac{T-1}{2}}},
\end{equation}
where $|\cdot|$ is the matrix determinant operator.

This approach makes the assumption that the location of the change point is known to be at $k$. In reality however the location of the change point is unknown, and the method can be extended to allow an unknown change point location by considering $\max_{1+\Delta\leq k \leq T-\Delta}\{\Lambda_k\}$. When the $Y_j$ are normally distributed then, for a fixed $k$, $-2\log(\Lambda_k)$  follows a chi-square distribution for large $T$ and for large $T-k$. Taking the maximum of $\Lambda_k$ over all possible $k$ results in a less tractable analytic distribution for the test statistic due to the necessity of correcting for multiple testing. For this reason, along with the fact that we do not wish to restrict ourselves to these distributional and asymptotic assumptions, we note that \eqref{LRteststat} can be easily adapted to the framework developed in Section \ref{ourmethod} by defining $d(k) = \Lambda_k$ and proceeding as usual. This suggests that different definitions of our matrix distance metric $d(k)$ can lead to substantially different results even in the same general framework. This idea is pursued further in Section \ref{normchoice}.

\subsection{Simulation based change point detection} \label{ourmethod}

It is of interest to establish a method of change point detection that does not require any distributional assumptions on $Y_j$, and we develop such a method in this section. Our approach is based on the bootstrap which offers a computational alternative that can well approximate the distribution of $Y_j$ through resampling. Under $H_0$, if the $Y_j$ are all independent and come from the same distribution $Y_j \sim f(\cdot \mid \Sigma)$ for all $1\leq j \leq T$, then bootstrapping the columns of $Y$ is appropriate. Though $f(\cdot \mid \Sigma)$ is unknown, we approximate it with the empirical distribution $\hat{f}(\cdot \mid \Sigma)$ which gives each observed column vector $Y_j$ an equal point mass of $1/T$. This is equivalent to resampling from the columns of $Y$ with replacement.

For many time series applications there may be autocorrelation present between the columns of $Y$. In this case resampling the columns of $Y$ would break the correlation structure and lead to bias in the approximation of the null distribution. To account for this autocorrelation in the resampling procedure, we use the sieve bootstrap \citep{buhlmann1997sieve}. In particular, for correlated data the $Y^{(b)}$ are generated for  autocorrelation of order $s$ by fitting the model 
\begin{equation}\label{sieveeq}
Y_j= \sum_{k=1}^s \hat{\phi}_k Y_{j-k} + \hat{\epsilon}_j
\end{equation}
 for each $j>s$. The $\hat{\phi}_k$ are estimated from the Yule-Walker equations, and used to solve for the $\hat{\epsilon}_j$s through equation \eqref{sieveeq}. The bootstrap residuals , $\hat{\epsilon}_j^{(b)}$, are resampled from all the $\hat{\epsilon}_j$s with replacement. This generates the bootstrapped $Y_j^{(b)}$ according to $Y_j^{(b)}= \sum_{k=1}^s \hat{\phi}_k Y_{j-k} + \hat{\epsilon}_j^{(b)}$.

Let $Y^{(b)}$ be one of the bootstrap resamples from $Y$, where each $Y_j^{(b)}$ are generated by bootstrapping from $\hat{f}(\cdot,\Sigma)$ in the case of independence or from the sieve bootstrap for correlated data. This is repeated for $b \in \{1,\dots,B\}$ where $B$ is the total number of bootstrap samples. $\Delta$ is a ``buffer'' that limits the change point detection from searching too close to the boundaries of data. We recommend $\Delta \approx n$. In the case where a change point location $k$ is closer than $\Delta$ to either $1$ or $T$, the change point will not be detected but, as will be seen in Section \ref{boundarysection}, these cases are near impossible to detect regardless of how small we make $\Delta$. For each $k \in \{1+\Delta,\dots,T-\Delta\}$,  $S^{(b)}(1,k)$, $S^{(b)}(k+1,T)$, and $d^{(b)}(k)$ are calculated where $S^{(b)}(i,j)=Y^{(b)}D(i,j)Y^{(b)T}$ and $d^{(b)}(k)=||S^{(b)}(1,k)-S^{(b)}(k+1,T)||_F$. Then $\hat{\mu}_{0}(k) = \frac{1}{B}\sum_{b=1}^B d^{(b)}(k)$ and $\hat{\sigma}_{0}^2(k)= \frac{1}{B-1}\sum_{b=1}^B (d^{(b)}(k)-\hat{\mu}_{0}(k))^2$ are calculated for each $k \in \{1+\Delta,\dots,T-\Delta\}$.

A z-score is then calculated for each potential change point $k \in \{1+\Delta,\dots,T-\Delta\}$ as

$$z^{(b)}(k) = \frac{d^{(b)}(k)-\hat{\mu}_{0}(k)}{\sqrt{\hat{\sigma}_{0}^2(k)}}$$

The change point occurs for the value of $k$ for which the z-score is largest, so we let $Z^{(b)} = \max_{k}\{z^{(b)}(k)\}$ for each bootstrap sample $b$. This is also performed on the observed data, with $z(k) = [d(k)-\hat{\mu}_{0}(k)]/\sqrt{\hat{\sigma}_{0}^2(k)}$ and $Z = \max_{k}\{z(k)\}$ being the test statistic. The corresponding p-value obtained from bootstrapping is

$$\mbox{p-value} = \frac{1}{B}\bigg|\left\{b: Z^{(b)} \geq Z\right\}\bigg|,$$
where $|\cdot|$ is the cardinality of the set. If the p-value is significant, i.e., if sufficiently few bootstrap replicates $Z^{(b)}$ exceed Z, then we reject $H_0$ and declare a change point exists for the value of $k$ with the highest z-score, i.e., at $\arg\max_k z(k)$.

It is also often the case that there exist more than one change point. In this case the data is split into two segments, one before the first change point and the other after it, and the bootstrap procedure is then repeated separately for each of the two segments. If a significant change point is found on a segment, then that segment is split in two again and this process is repeated until no more statistically significant change points are found. In practice, this procedure terminates after a small number of rounds because each iteration on average halves the amount of data which greatly reduces power to detect a change point after each subsequent iteration.

\subsection{Difficulty of detection near the boundary} \label{boundarysection}

We alluded above to the difficulty of detecting change points near the boundaries of the data, and will now investigate this issue in more detail. When a change point $k$ is very close $1$, then the empirical covariance matrix $S(1,k)$ is constructed using a very small amount of data and its estimate is unstable with high variance. Similarly, when $k$ is very close to $T$, $S(k+1,T)$ suffers from the same problem. This makes change point detection hard: if the empirical covariance matrix is highly variable, the noise from the estimation of the covariance matrices can make any possible differences between $\Sigma_1$ and $\Sigma_2$ statistically difficult to detect. 

In an attempt to quantify just how difficult of a problem change point detection is near the boundary, we find the analytic form of $E[d(k)]$ under $H_0$ in the case of normally distributed $Y_{ij}$.

\newtheorem*{expval}{Theorem 1}
\begin{expval}
Let $Y_j \sim MVN(0,\Sigma)$ for $j \in \{1,...,T\}$ all i.i.d., then for $d(k) = ||S(1,k)-S(k+1,T)||_F^2$ for any $k \in \{2,...,T-1\}$ we have
\begin{equation} \label{expeq}
E[d(k)] = 	\left(\frac{1}{k} + \frac{1}{T-k}\right)(\tr(\Sigma^2)+\tr(\Sigma)^2)
\end{equation} 
\end{expval}
\begin{proof}
We have that $d(k) = 	\tr\{[S(1,k)-S(k+1,T)]^2\} = \tr\{[Y^TY(D(1,k)-D(k+1,T))]^2\} = \sum_{i=1}^T\sum_{j=1}^T (Y_i^TY_j)^2C_{jj}C_{ii}$ where $C = D(1,k)-D(k+1,T)$.
From the variance of a Gaussian quadratic form we have that $E[(Y_i^TY_i)^2] = 2\tr(\Sigma^2)+\tr(\Sigma)^2$. Similarly for the covariance case when $i \neq j$, we have that $E[(Y_i^TY_j)^2)] = E[E[(Y_i^TY_j)^2|Y_j]] = E[Y_j^T\Sigma Y_j] = \tr(\Sigma^2)$. These combine to give us the result. For the more detailed algebra expanded upon, see Appendix \ref{expeqappendix}. 
\end{proof}

The implication of the Theorem (Equation \eqref{expeq}) is that the expected difference asymptotes to infinity as $k$ approaches $0$ or $T$ and is minimized when $k=\lfloor T/2\rfloor$. Although this result assumes multivariate normal data, we expect that the qualitative nature of the result generalizes beyond the normal distribution. The increase in $E[d(k)]$ is confirmed through simulation under $H_0$ and demonstrated in Figure \ref{edktheoVSemp}. The implication is that the noise in the estimation of the covariance matrices on both sides of the change point is minimized when both $S(0,k)$ and $S(k+1,T)$ have sufficient data for their estimation. When $k$ is close to $0$, then even though $S(k+1,T)$ has low variability, the large increase in variability of $S(0,k)$ leads to an overall noisier outcome. This demonstrates that the strength of the method is only as strong as its weakest estimate. For the purposes of study design and data collection, if we suspect that a change point occurs at a certain location, perhaps for theoretical reasons or based on past studies, we need to ensure that there is sufficient data collected both before and after the suspected change point if we are to have any hope of detecting it.

\begin{figure}[!ht] 
\centerline{\includegraphics[scale=.7,trim=0 0.6cm 0 1.5cm,clip=true]{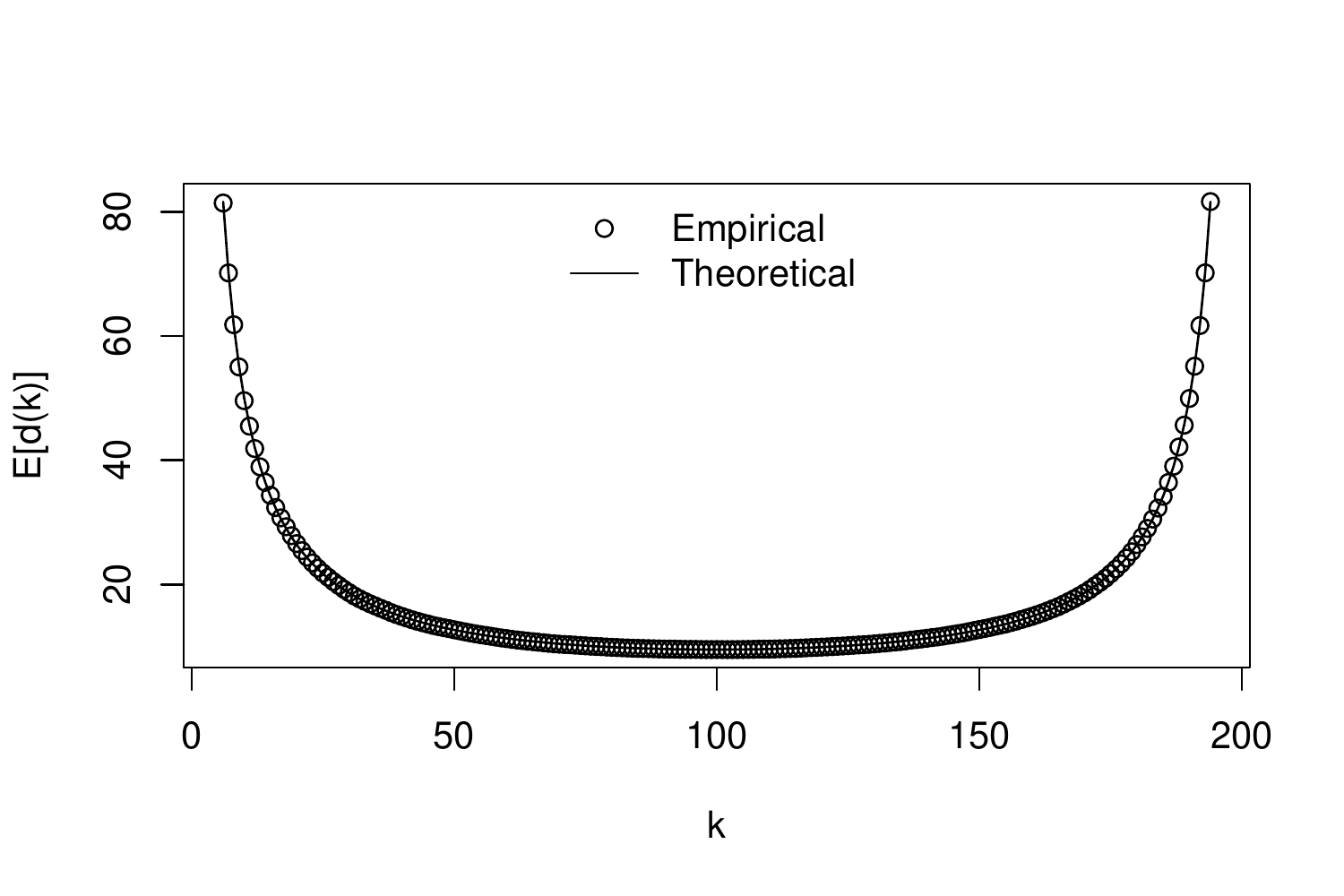}}
\caption{Difficulty of change point detection near the boundaries of data. With $T=200$, $n=20$, and $\Sigma=I_n$ (the identity matrix of order $n$), for each potential change point $1< k<  T$, we estimate $E[d(k)]$ by averaging $d(k)$ over $10000$ simulations under $H_0$ and show the location of expected values with markers. These empirical estimates are contrasted with the theoretical expectation, shown as a solid line, given by Theorem 1, Equation \eqref{expeq}.} \label{edktheoVSemp}
\end{figure}

\section{Simulation}

\subsection{The relationship between $T$ and $n$ for statistical power} \label{TVSn}

Estimation of the covariance matrix requires $T$ to be large relative to $n$ because the empirical covariance matrix has $n(n-1)/2$ elements that need to be estimated, so there is high variability in estimates if $T$ is small. If $T$ is too small, then even if a change point exists, the empirical covariance matrix may be so variable that the change point is undetectable. This problem is exacerbated when trying to detect change points near the boundary as discussed in Section \ref{boundarysection}.

While it is intuitive that $T$ needs to grow as some function of $n$ in order to maintain any reasonable statistical power to detect change points, it is unclear what that function of $n$ is. We investigate here further, using simulation, at what rate the number of longitudinal observations $T$ needs to grow with system size $n$ in order to maintain the same statistical power. We consider the case where a single change point occurs at the midpoint $\lfloor T/2\rfloor$, and $Y_j \sim MVN(0,\Sigma_1)$ for $j\leq T/2$ and $Y_j \sim MVN(0,\Sigma_2)$ for $j> T/2$ where:
$$
\Sigma_1=I_n \hspace{5pc}
\Sigma_2=\left[
\begin{array}{cc}
\underbrace{(1-\rho)I+\rho 11^T}_{4\times 4}	&	\underbrace{0}_{4 \times (n-4)}	\\
\underbrace{0}_{(n-4)\times 4}	&	\underbrace{I}_{(n-4) \times (n-4)}	\\
\end{array}
\right]
$$
where $\rho=0.9$. In other words, $\Sigma_2$ is a block or partitioned matrix with exchangeable correlation within the blocks on the diagonal, and 0s in the off-diagonal blocks. We simulate instances of $Y$ in this fashion 10000 times for each of $n=4,8,12$.

In Figure \ref{cpcombinedprobs} we compare the performance of the method, measured by the proportion of the 10000 iterations resulting in a statistically significant change point, described in Section \ref{ourmethod} for change point detection for the three different values of $n$. The asymmetry in Figure \ref{cpcombinedprobs} around the true change point is caused by having $\Sigma_1$ first followed by $\Sigma_2$. If the order of $\Sigma_1$ and $\Sigma_2$ is reversed, then the asymmetry will be reversed as well. We find that the probability of detecting the correct change point is the same for all $n$ if we increase $T$ by a quadratic rate in $n$ as $T(n) = n(n-1)+C$ for the constant $C$, a functional form we discovered by numerical exploration. In our simulations we considered the $\alpha=0.05$ significance level and $C=30$.  The intuition behind a quadratic rate is that as $n$ increases, the number of entries in the empirical covariance matrix increases quadratically and therefore the noise in the Frobenius norm increases quadratically. Increasing $T$ quadratically with $n$ appears to balance out the added noise for increasing the dimensions of the correlation matrices, and stabilizes the statistical power to detect the change point.  We would therefore recommend that if one wants to increase $n$, then there needs to be an associated increase in the number of observations that is quadratic in $n$ in order to retain the ability to detect a change point with the same power.

\begin{figure}[!ht]
\centerline{\includegraphics[scale=.9,trim=0 0.6cm 0 1.5cm,clip=true]{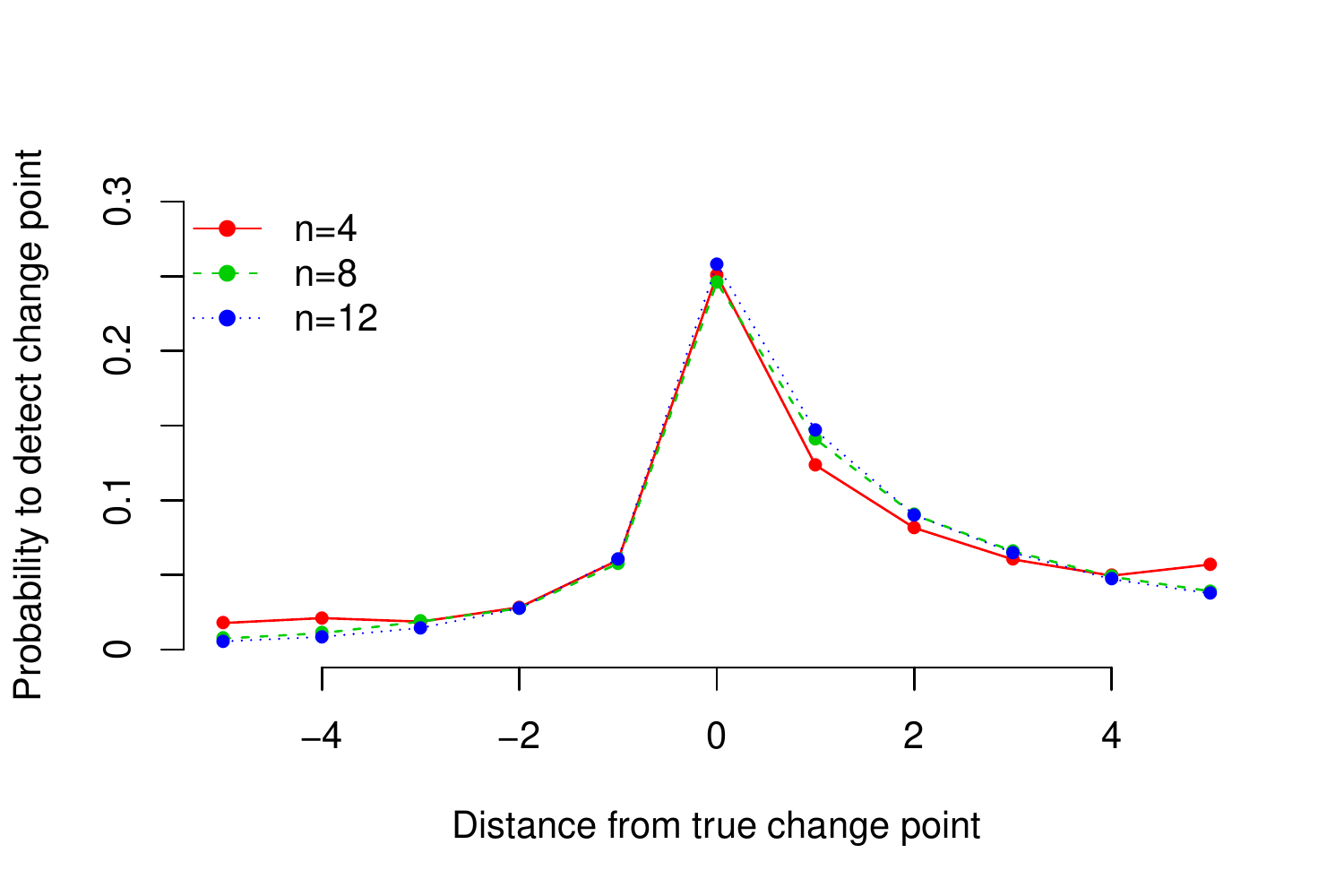}}
\caption{Change point detection statistical power as a function of $n$ and $T$. The $y$ axis represents the probability that a change point is detected at a particular time point. The $x$ axis is the distance of a time point in either direction from the true change point. These probabilities are estimated based on 10000 iterations for each $n$.} \label{cpcombinedprobs}
\end{figure}

\subsection{Comparison of different matrix norms} \label{normchoice}

Up until this point our proposed method has dealt with taking the Frobenius norm of the difference of empirical correlation matrices. The choice of the Frobenius norm was simply for algebraic simplicity of Theorem 1 (Equation \eqref{expeq}). Though it is more simple than many other matrix norms for such calculations, there is no reason to believe that the Frobenius norm is uniformly the best choice of matrix norm if the objective is to maximize the statistical power of change point detection. There may be some change points that the Frobenius norm is good at detecting, but there may be other change points for which a different matrix norm or distance metric would be more suitable. We investigate this question more closely in this Section.

Because the Frobenius norm sums the squared entries of a matrix, it is intuitive to expect that the Frobenius norm would be ideal for detecting change points in systems that demonstrate large-scale, network-wide changes in the correlation pattern. On the other hand, the Frobenius norm likely would not be very powerful in detecting small-scale local changes in correlation network structure. The rationale for this argument is that by summing over all the changes in the network structure, if there are very few changes relative to the entire network, then the Frobenius norm would be dominated by noise from the largely unchanged matrix elements.

We consider a different matrix norm, the Maximum norm, that is appealing for the case of small-scale, local changes. The Maximum norm of a matrix is simply the largest element of the matrix in absolute value. Intuitively, this norm would be ideal if there was just a single, but very large, change in the covariance matrix. If only one element of the covariance matrix changes, but the change is quite large, the Maximum norm would still be able to detect this change. Here the Frobenius norm would likely fail due to the sum of the all the changes being dominated by noise. The likelihood ratio test in Equation \eqref{LRteststat} is more similar to the Frobenius norm than the Maximum norm in that it utilizes all entries in the covariance matrix rather than using only one element. As a result, we may expect the likelihood-ratio distance metric to be more similar to the Frobenius norm in performance than it is to the Maximum norm.

We compare the Frobenius norm, the Maximum norm, and likelihood-ratio in Equation \eqref{LRteststat} through simulation with varying proportion of the network altered at the change point. To do this, we generated $Y_j$ from a multivariate normal distribution with $T=400$ and a single change point occurring at $t=200$. Prior to the change point $Y_j \sim MVN(0,\Sigma_1)$ for $j \leq t$ and after the change point  $Y_j \sim MVN(0,\Sigma_2)$ for $j > t$, where we modify the dimension of the upper-left block of $\Sigma_2$ to change the proportion of the network that is altered at the change point. In each case $\rho$ is selected such that the change point is detected with $50\%$ power using the Frobenius norm. This provides a reference for how the Frobenius norm compares with the Maximum norm and the likelihood-ratio. The results are displayed in Figure \ref{comparenorms}, which confirms our intuition. When a small proportion of the network is altered, the Maximum norm is more powerful at change detection than the Frobenius norm and likelihood ratio metric. When a large proportion of the network is altered at the change point, then the Frobenius norm and likelihood ratio metric are more powerful. While the likelihood ratio metric is more similar to the Frobenius norm than it is to the Maximum norm, it is still less sensitive to wide-spread subtle network changes than the Frobenius norm.

\begin{figure}[!ht] 
\centerline{\includegraphics[scale=.8,trim=0 0.6cm 0 1.5cm,clip=true]{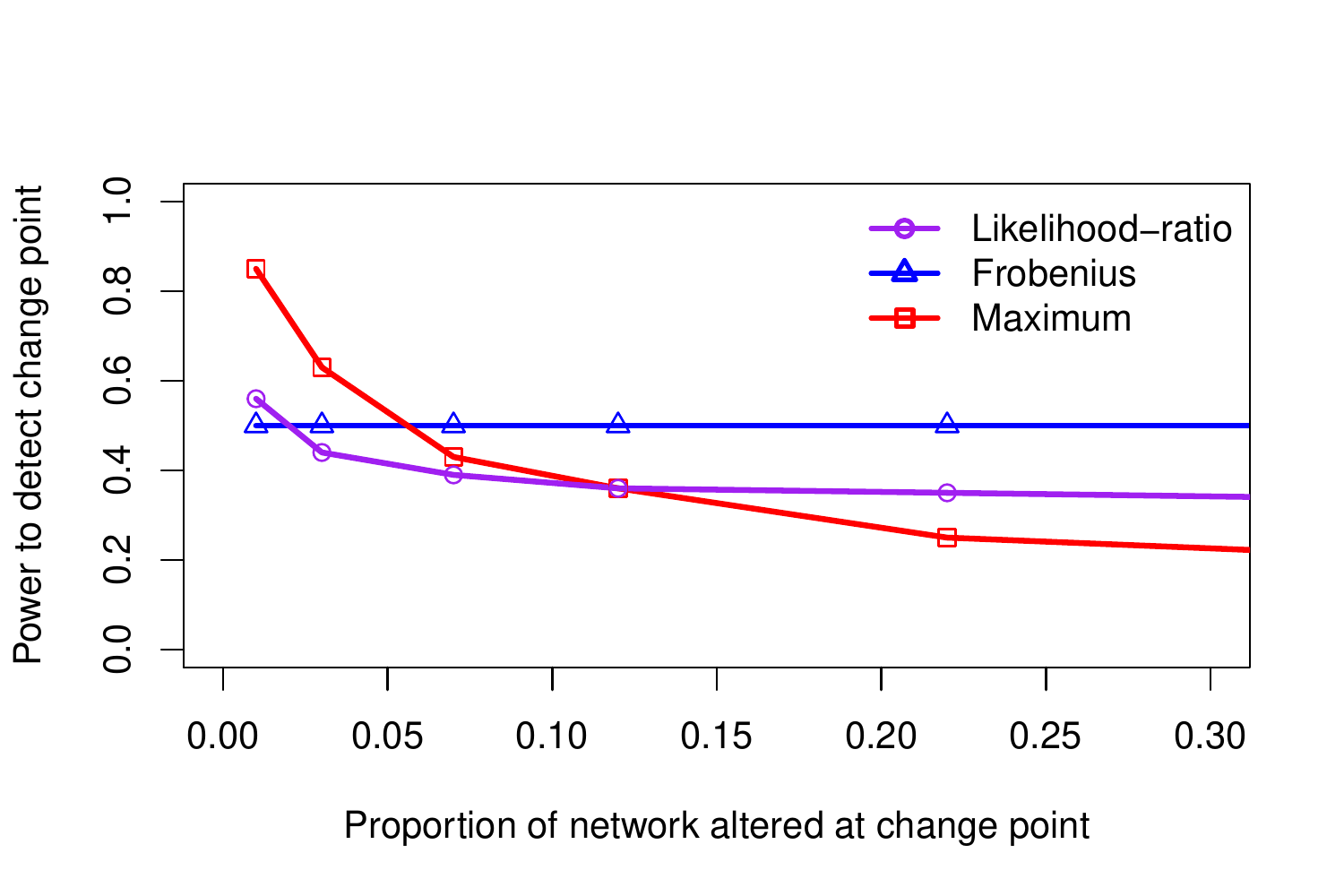}}
\caption{Power comparison for different matrix norms. For each point on the $x$-axis, a value of $\rho$ in the definition of $\Sigma_2$ is selected such that the Frobenius norm has $50\%$ power to detect a change point. As the proportion of the network altered at the change point increases, we adjust the value of $\rho$ correspondingly.} \label{comparenorms}
\end{figure}

These considerations naturally lead to the following question: which norm or metric should be used? The answer clearly depends on the anticipated nature of the change point, and is therefore difficult because often the nature of the change point is unknown. In fact, change point detection is used even when one is not sure a change point exists. If there is some \textit{a priori} knowledge of a type of change point perhaps specific to the problem at hand, then that information could be used to select an appropriate norm. For example, suppose we investigate a network constructed from stock return correlations and the time period under investigation happens to encompass a sudden economic recession. The moment the recession strikes, it is likely that there will be large-scale changes in the underlying network, and therefore the Frobenius norm might be a good choice.

One important issue that deserves emphasis is \textit{when} the choice of the norm to use should be made. It is very important that the choice of norm is made prior to looking at the data. If the analysis is performed multiple times repeatedly with different choices for the matrix norm, and the norm with the ``best'' results is selected, this would be deeply flawed and would invalidate the interpretation of the p-value. See \citet{gelman2013garden} for an informative discussion of the problem of inflated false positive rates that result when the choice of the specific statistical procedure to use, in this case the norm, is not made prior to all data analysis.

\subsection{Detecting multiple change points} \label{multiplesims}

It may be the case that more than one change point occurs in the data. In this case, the method described in Section \ref{ourmethod} can still be applied to search for additional change points by splitting the data into two segments at the first significant change point and then repeating the procedure on each segment separately.  This process is repeated recursively on segments split around significant change points until no additional statistically significant change points remain. Each test is performed at the $\alpha=0.05$ level (or at another user-specifid level). Though multiple comparisons may seem like a potential problem here, in fact there is no problem because further tests are only performed conditional on the previous change points being elected as statistically significant. This prevents the false positive rate from being inflated. A sliding window approach can also be used to detect multiple change points as is done in \citet{hawkins2003changepoint} and \citet{peel2014detecting}.

We investigate the performance of our method for detecting multiple change points through simulation. Consider the case where $T=400$ time points are observed for $n=10$ nodes in a network. Data follows a multivariate normal distribution with mean 0 and covariance $\Sigma_1$ for $1\leq t\leq 100$ and for $201 \leq t \leq 300$, but has covariance  $\Sigma_2$ for $101\leq t\leq 200$ and for $301 \leq t \leq 400$. We define $\Sigma_1$ and $\Sigma_2$ as in Section \ref{TVSn} with $\rho=0.9$, except that the upper left block of $\Sigma_2$ is $5 \times 5$ in this case. The probability that a change point is detected at each particular location is estimated from 10000 iterations and is shown in Figure \ref{multipleCPfig}. Statistical significance is assessed at the usual $\alpha=0.05$ significance level. We use the Frobenius norm here because changing between $\Sigma_1$ and $\Sigma_2$ constitutes large-scale change in the network.

\begin{figure}[!ht] 
\centerline{\includegraphics[scale=.8,trim=4cm 3.48cm 4cm 4.3cm,clip=true]{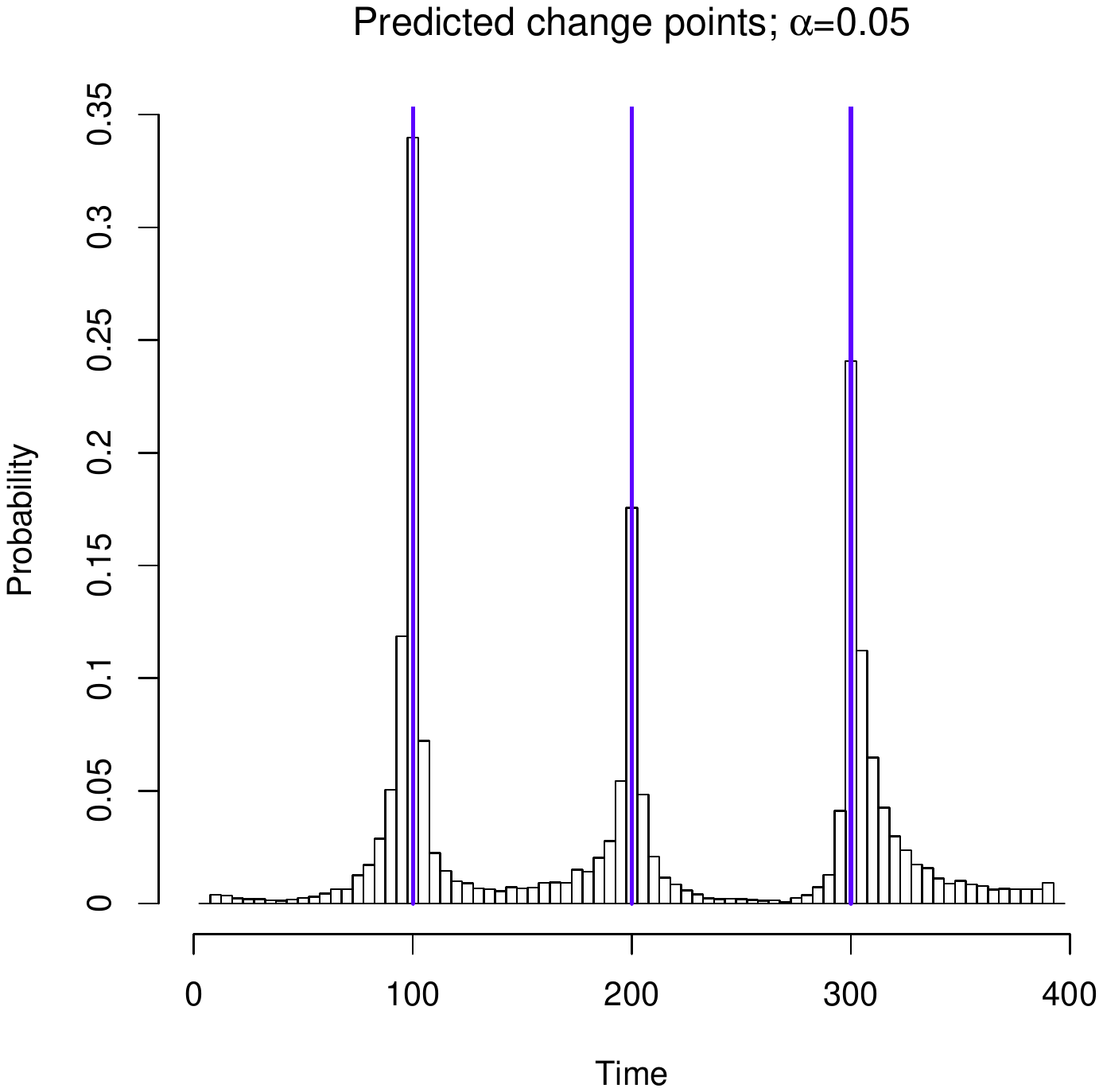}}
\caption{Detection of multiple change points. The $y$ axis represents the probability a change point is detected in a bin, each containing five adjacent time points, based on 10000 simulations total. The true locations of the change points are marked with vertical blue lines. } \label{multipleCPfig}
\end{figure}

As expected, the closer a location is to a change point, the more likely that location is found to be a statistically significant change point. However, there is an asymmetry in the ability to detect the different change points. The change point at $t=200$ is more difficult to detect than the change points at $t=100$ and $t=300$. This is because if we consider all $400$ data points and split them around $t=200$, the two resulting covariance matrices on either side of $t=200$ are expected to be the same because for both sets of $200$ observations, half are from $\Sigma_1$ and half are from $\Sigma_2$. This means we have almost no statistical power to detect a change point in the first iteration at $t=200$. Instead, the change points at $t=100$ and $t=300$ are picked up first. The reason the change point at $t=100$ is easier to detect than the one at $t=300$, despite each being equally far from its respective boundary, is because the data is ordered with the first $100$ observations generated from $\Sigma_1$ and the last $100$ from $\Sigma_2$. If this is reversed, then $t=300$ becomes the change point most likely to be detected. After the first change point is detected, power is reduced for the remaining change points due to the reduction in sample size that occurs due to dividing the data into smaller segments.

\section{Data Analysis}

\subsection{Correlation networks of stock returns}

Our first data analysis example deals with networks constructed from correlations of stock returns. Networks constructed from correlations of stock returns have been used in the past to investigate the correlation structure of markets as well as to detect changes in their structure \citep{mantegna1999hierarchical,onnela2003dynamic,onnela2003dynamics}. Here we use a data set first analyzed in \citet{onnela2004clustering} and apply our change point detection methods to it. 

A total of $n=114$ S\&P 500 stocks were followed from the beginning of 1982 to the end of 2000, keeping track of the stock price at closing for $T=4786$ trading days over that time period. This data is publicly available and had been gathered for analysis previously where correlation networks were constructed based on the correlation between log returns in moving time windows \citep{onnela2003dynamics}. If the price of the $i$th stock on the $j$th day is $P_{ij}$, then the corresponding log return is $R_{ij}=\log(P_{ij})-\log(P_{i,j-1)})$. The log returns did not demonstrate statistically significant autocorrelation (Durbin-Watson p-value of $0.11$) so the independent bootstrap was used.

Given that the stock market evolves constantly, and given the long time interval in the observed data, it may not be safe to assume that the correlation between the log returns of any two stocks stays fixed over time. If a correlation network were constructed by assuming edges between every two stocks with correlation greater than some threshold, and if all of the 19 years of data were used at once, the resulting network would likely be an inaccurate representation of the market if in fact the true underlying network changes with time \citep{onnela2006complex}. A more principled approach would be to first test for change points in the correlation network and then build multiple networks around those change points if necessary.

\begin{figure}[!ht] 
\centerline{\includegraphics[scale=.9,trim=0 0 0 0,clip=true]{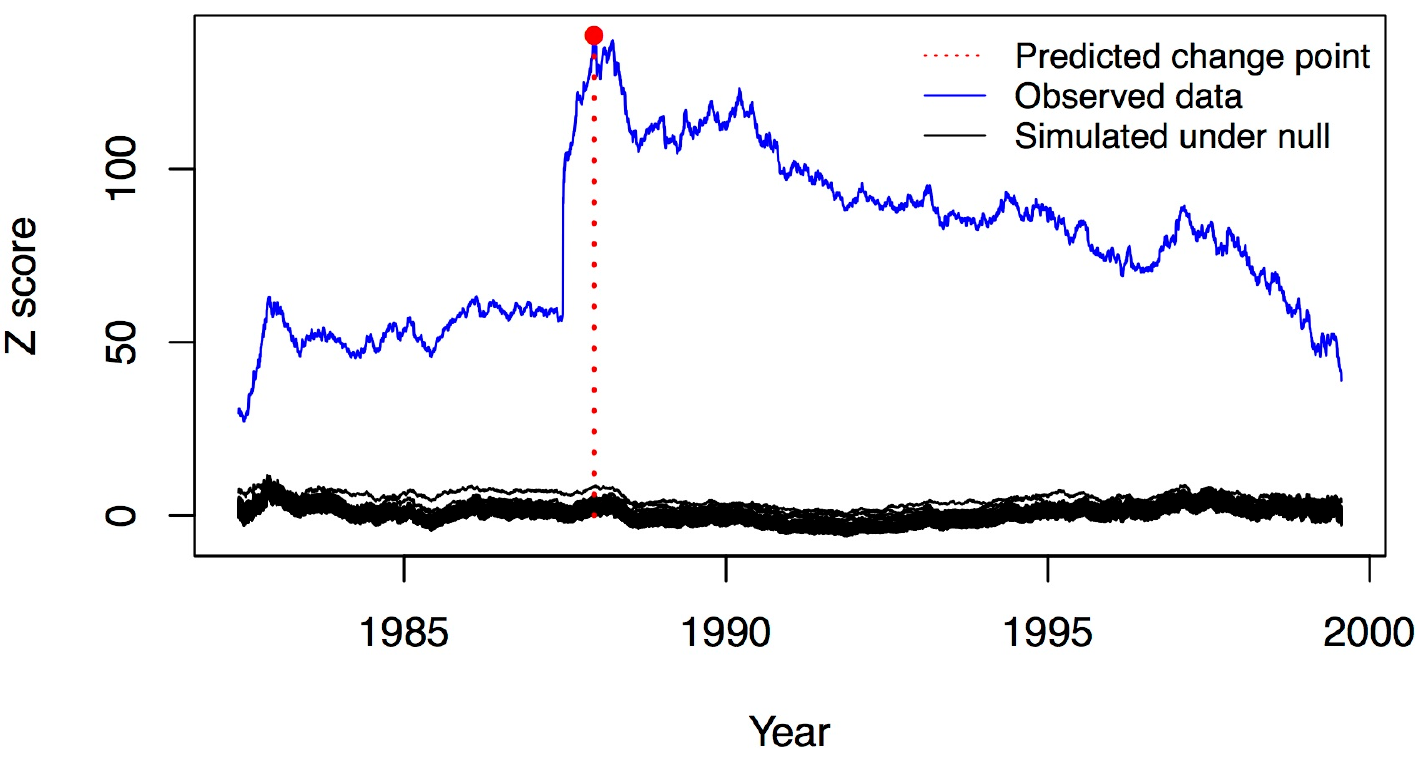}}
\caption{Change point detection in stock returns. With $n=116$ stocks tracked over $T=4786$ days ($\sim 19$ years), the blue line is the empirical z-score $z(k)$ while the clustered black lines are the z-scores simulated under $H_0$ using bootstrap as defined in Section \ref{ourmethod}. A significant change point is detected near the end of the year $1987$ corresponding to the well documented crash at the end of that year.} \label{StockPriceFig}
\end{figure}

Following the method proposed in Section \ref{ourmethod}, the stock price data is resampled 500 times assuming the null hypothesis of no change points, and the observed data is compared with these simulations to determine if and where a change point occurs. These simulation results are displayed in Figure \ref{StockPriceFig}. There is strong statistical evidence of a change point at the end of the year 1987 evidenced by a p-value $< 0.002$. The sieve bootstrap approach arrives at the same result, confirming the lack of temporal autocorrelation in log returns. We used the Frobenius norm as our goal was to find events that could lead to large-scale shocks to the correlation network. There were several other significant change points, but we focus on the first and most significant change point here. The stock market crash of October 1987, known as ``Black Monday'', coincides with the first detected change point \citep{onnela2003dynamic}. The stock market crash evidently drastically changed the relationship between many of the stocks leading to a stark change in the correlation network. For this reason it is advisable to consider the network of stocks before and after the stock market crash separately, as well as splitting the data further around potential additional significant change points, rather than lumping all of the data together to construct a single correlation network.

\subsection{Correlation networks of fMRI activity } \label{fmrisection}

Our second data analysis example deals with networks constructed from correlations in fMRI activity in the human brain. The Center for Cognitive Brain Imaging at Carnegie Mellon University collected fMRI data as part of the star/plus experiment for six individuals as they each completed a set of 40 trials \citep{mitchell2004learning}. Each trial took approximately 27 seconds to complete. The subjects were positioned inside an MRI scanner, and at the start of a trial, each subject was shown a picture for four seconds before it was replaced by a blank screen for another four seconds. Then a sentence making a statement about the picture just shown was displayed, such as ``The plus sign is above the star,'' and the subject had four seconds to press a button ``yes'' or ``no'' depending on whether or not the sentence was in agreement with the picture. After this the subject had an interstimulus period of no activity for 15 seconds until the end of the trial. We avoid referring to this as ``resting state'' due the reserved meaning of that label for extended periods of brain inactivity. Trials were repeated with different variations, such as the picture being presented first before the sentence, or with the sentence contradicting the picture. MRI images were recorded every $0.5$ seconds, for a total of about 54 images over the course of a trial, corresponding to a total of $40\times 54=2160$ images total. Each image was partitioned into 4698 voxels of width 3mm. The study data are publicly available \citep{cmufMRIwebsite}.


If we were to analyze a single trial, change point detection would be quite difficult for the data in its raw form  for $n= 4698$ voxels which is very large compared with the number of data points $T=54$. Any empirical covariance matrix for these values of $n$ and $T$ would be too noisy to detect any statistically significant change point. We therefore combine our analysis on the eight trials  where the picture is presented first and the sentence agrees with the picture for all six individuals. To accommodate the repeated trials in our correlation estimates, we define $S_{lk}(i,j)$ to be the covariance estimator in Equation \eqref{empcor} for the $l$th individual and $k$th trial only. The resulting covariance matrix averaged over all the trials and individuals is

$$S^*(i,j) = \sum_{l=1}^6\sum_{k=1}^8 S_k(i,j)/48$$

Change point detection is then performed as before except using the $S^*(i,j)$ instead of the usual $S(i,j)$. Even though we have effectively increased the amount of data 48-fold by combining multiple trials and individuals together, the number of observed data points is still far fewer than the $n=4698$ voxels. To reduce the number of nodes to a manageable size, we group the voxels into 24 distinct regions of interest (ROIs) in the brain following \citet{hutchinson2009modeling}, and we average the signals over all voxels within the same ROI. With 24 nodes and $54 \times 48 = 2592$ data points, empirical covariance matrices can be estimated with sufficient accuracy to detect change points in the network of ROIs so long as the change points occur sufficiently far from the beginning or end of the trial (see Section \ref{boundarysection}). Under the assumption that the network is drastically different when comparing the interstimulus state to the active state, we use the Frobenius norm. For each of the 6 individuals there was a significant presence of autocorrelation (Durbin-Watson p-value $<0.001$), so the sieve bootstrap is used for inference. First order autocorrelation ($s=1$) was used as it minimized mean squared error in cross-validation.

The most significant change point occurred $t=12$ seconds into the trial, though it was not statistically significant (p-value of $0.18$). This indicates that the network of interactions in the first part of the trial when the subject is actively reading, visualizing, responding, and connecting stimuli is most different from the interstimulus portion of the trial, though the difference is not statistically significant. If autocorrelation is ignored and instead independence is assumed, then the same change point appears significant with a p-value less than $0.001$. This example demonstrates how ignoring autocorrelation in the data can lead to an inflation of the false positive rate.

\begin{figure}[t!] 
\centerline{\includegraphics[scale=1,trim=0 1cm 0 1cm ,clip=true]{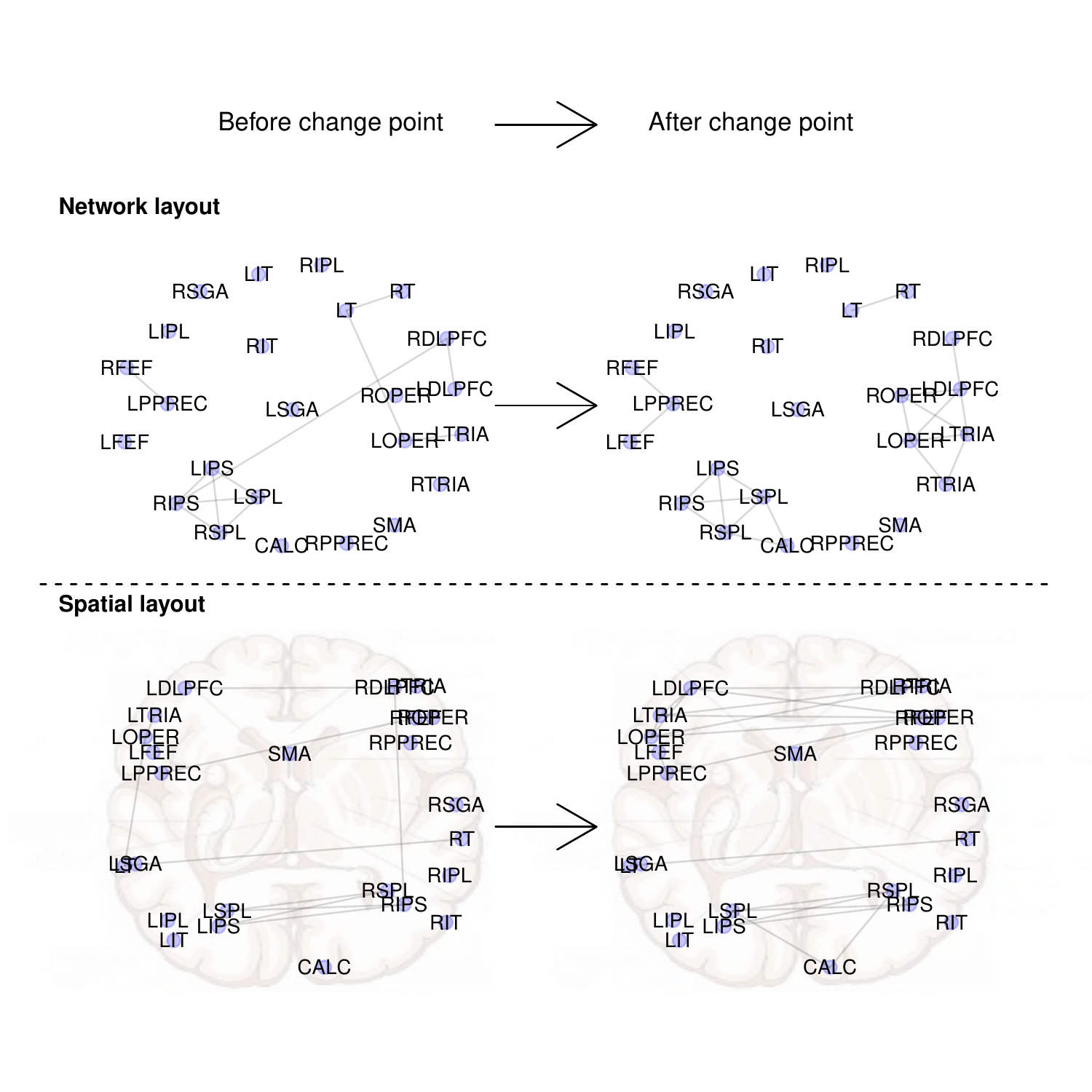}}
\caption{Brain region of interest (ROI) networks before and after the most likely change point. The network transition around the most likely change point are displayed with two different layouts. In the top panel, nodes are positioned to best display the pre-change point network topology. Those same node positions are used in the post-change point network in the top right. The networks on the bottom have nodes positioned according to their Talaraich coordinates \citep{lancaster2000automated} that accurately represent their anatomical location in the brain.} \label{fMRInetworksfig}
\end{figure}

Though we fail to reject the null hypothesis and find no statistically significant change point at the $\alpha=0.05$ level, we examine the position of the most likely change point, if one exists, at $t=12$. We construct two networks between the ROIs, one from $S^*(1,24)$ corresponding to the first 12 seconds of the trial (recall that $i$ and $j$ in $S^*(i,j)$ index time points that are $0.5$ seconds apart) and one from $S^*(25,54)$ corresponding to the remaining 15 seconds of inactivity in the trial. The networks are constructed such that an edge is shown between two nodes if and only if their pairwise correlation is greater than 0.5 in absolute value. The two networks are displayed in Figure \ref{fMRInetworksfig}. The correlation threshold to determine if an edge is present was selected such that the network after the change point had $20$ edges, and the same threshold was used for both before and after networks. During the inactive period after $t=12$ there is an increase in connectivity in the network. An explanation for why the most likely change point occurs at $t=12$ could be because behavior after the change point corresponds with constant inactivity, whereas before the change point there is a mixture of inactivity, thinking, decisions, and other mental activity likely taking place. These different mental activities could dampen the observed correlations when averaged all together. Given the many different tasks that occur over the course of the trial, there almost certainly exist numerous changes in the correlation structure of brain activity. The lack of statistical significance in this example helps to illustrate the phenomenon discussed in section \ref{TVSn}; even if change points exist, as is almost certainly the case here due to the numerous mental tasks required over the course of a trial, there is no hope of discovering them unless the number of observations is far greater than the size of the network ($T \gg n$).

\section{Discussion}

In this paper, an existing change point detection method was adapted to correlation networks using a computational framework. Many past treatments of change point detection make a distributional assumption on the observed characteristics, but our framework utilizes the bootstrap in order to avoid this restriction. Traditional methods also assume independence between observations and upon first glance this assumption seems unreasonable. For instance, consider the stock market data. Stock prices are often modeled as a Markov process, which implies a strong autocorrelation between consecutive observations. For this reason the stock prices themselves cannot be used as input for our algorithm, but rather the log returns are used. Similar to the random noise in a Markov chain being independent, the assumption of the returns being independent is more reasonable, as was also demonstrated by the Watson-Durbin test. The fMRI voxel intensities, however, demonstrated significant autocorrelation and required the sieve bootstrap procedure.

We extended our framework to allow for multiple change points. If the first change point is found to be statistically significant, then the data is split into two parts on either side of the change point and the algorithm is repeated for each subset. This process of splitting the data around significant change points continues until there are no more significant change points. The fMRI data analysis in Section \ref{fmrisection} found no significant change points but, due to the many changes in stimuli, there are likely multiple points in time where the structure of interaction between regions of interest in the brain changes. This negative result could likely be remedied by collecting higher temporal frequency imaging data. With each split of the data, $T$ is approximately halved while $n$ remains the same and, as shown in Section \ref{TVSn}, this further lowers the power to detect change points. As long as the data have sufficient temporal resolution, it should be possible to use the proposed framework to detect multiple change points in correlation networks across different domains.

\section{Appendix: Proof of Theorem 1}

\subsection{Proof of Theorem 1} \label{expeqappendix}
The proof starts by writing the Frobenius norm as a sum over all pairs of observations, and then takes expectation, making use of what we know of the first two moments of a quadratic form of normally distributed variables. The following is the derivation of $E[d(k)]$ under $H_0$ for normally distributed observations:
\[\begin{array}{rl}
d(k) =& \tr\{[S(1,k)-S(k+1,T)]^T[S(1,k)-S(k+1,T)]\}	\\
	=&	\tr\{[S(1,k)-S(k+1,T)]^2\}	\\ 	
	=& \tr\{(Y^TY\underbrace{(D(1,k)-D(k+1,T))}_{C})^2\}	\\
	=& \sum_{i=1}^T\sum_{j=1}^T (Y_i^TY_j)^2C_{jj}C_{ii}	\\
	=& \sum_{i=1}^T\left\{(Y_i^TY_i)^2C_{ii}^2+\sum_{j \in \{1,...,i-1,i+1,...,T\}}(Y_i^TY_j)^2C_{ii}C_{jj}\right\}	\\
\end{array}\]
Taking expectation gives

\[\begin{array}{rl}
E[d(k)]	=&	\sum_{i=1}^T\left\{E[(Y_i^TY_i)^2]C_{ii}^2+\sum_{j \in \{1,...,i-1,i+1,...,T\}}E[(Y_i^TY_j)^2]C_{ii}C_{jj}\right\}	\\
	=& \left(\frac{1}{k}+\frac{1}{T-k}\right)E[(Y_i^TY_i)^2] + \sum_{i=1}^k \left(\frac{k-1}{k^2}-\frac{1}{k}\right)E[(Y_i^TY_j)^2] \\
	&+ \sum_{i=k+1}^T \left(\frac{T-k-1}{(T-k)^2}-\frac{1}{T-k}\right)E[(Y_i^TY_j)^2]	\\
	=& \left(\frac{1}{k}+\frac{1}{T-k}\right)E[(Y_i^TY_i)^2] + k \left(\frac{k-1}{k^2}-\frac{1}{k}\right)E[(Y_i^TY_j)^2]	\\
	& +(T-k)\left(\frac{T-k-1}{(T-k)^2}-\frac{1}{T-k}\right)E[(Y_i^TY_j)^2]	\\
	=& \left(\frac{1}{k}+\frac{1}{T-k}\right)E[(Y_i^TY_i)^2] + \left(\frac{k-1}{k}-1+\frac{T-k-1}{T-k}-1\right)E[(Y_i^TY_j)^2]	\\
	=& \left(\frac{1}{k}+\frac{1}{T-k}\right)(2\tr(\Sigma^2)+\tr(\Sigma)^2) + \left(\frac{k-1}{k}+\frac{T-k-1}{T-k}-2\right)\tr(\Sigma^2)	\\
	=& \left(\frac{1}{k} + \frac{1}{T-k}\right)(\tr(\Sigma^2)+\tr(\Sigma)^2)	\\
\end{array}\]
The last line is the result of Theorem 1 in Equation \eqref{expeq}.

\bibliography{CPbib}

\end{document}